%% file: ms.tex
\documentclass[11pt]{article}

\usepackage[utf8]{inputenc}

\usepackage{amsmath,amsfonts,amsthm,amssymb,color}
\usepackage[usenames,dvipsnames,svgnames,table]{xcolor}
\definecolor{darkgreen}{rgb}{0.0,0,0.9}
\usepackage{mathtools}
\usepackage{authblk}
\usepackage{fullpage}
\usepackage{parskip}
\usepackage{comment}
\usepackage{tikz}
\usepackage{bbm}
\usepackage{dsfont}
\usepackage[sc]{mathpazo}
\usepackage[basic]{complexity}
\usepackage{algorithm2e}
\usepackage[colorlinks=true,
citecolor=OliveGreen,linkcolor=BrickRed,urlcolor=BrickRed,pdfstartview=FitH]{hyperref}
\usepackage[capitalize,nameinlink]{cleveref}
\usepackage{tcolorbox}
\usepackage[short]{optidef}

\newtcolorbox{wbox}
{
	colback  = white,
}

\SetKwInOut{Input}{Input}
\SetKwInOut{Output}{Output}
\SetKwFunction{Uncross}{\textsc{Uncross}}
\SetKwFunction{MergeUncross}{\textsc{MergeUncross}}
\SetKwFunction{PerfectMatching}{\textsc{PerfectMatching}}
\SetKwBlock{InParallel}{in parallel do}{end}
\SetKwFor{ParallelFor}{for}{in parallel do}{end}
\newcommand*{\Qplus}{\mathbb{Q_+}}

\newcommand*{\suppress}[1]{}

\newcommand*{\norm}[1]{\|#1\|}

\newcommand*{\cN}{\mathcal{N}}

\newcommand*{\cR}{\mathcal{R}}

\makeatletter
\def\thm@space@setup{%
	\thm@preskip= 10pt
	\thm@postskip=\thm@preskip % or whatever, if you don't want them to be equal
}
\makeatother

\makeatletter
\renewcommand{\paragraph}{%
	\@startsection{paragraph}{4}%
	{\z@}{5pt}{-1em}%
	{\normalfont\normalsize\bfseries}%
}
\makeatother

\newtheorem{theorem}{Theorem}
\newtheorem{lemma}{Lemma}

\newtheorem{note}{Note}

\theoremstyle{definition}
\newtheorem{definition}{Definition}

\newtheorem{remark}{Remark}

\newtheorem{alg}{Algorithm}

\newenvironment{fminipage}%
{\begin{Sbox}\begin{minipage}}%
		{\end{minipage}\end{Sbox}\fbox{\TheSbox}}

\def\norm#1{\left\| #1 \right\|}

\newcommand{\CR}{\mbox{${\cal R}$}}

\newcommand{\worth}{\mbox{\rm worth}}
\newcommand{\surplus}{\mbox{\rm surplus}}

\newcommand{\mn}{\mbox{\textsc{min}}}

\title{Arctic Auctions, Linear Fisher Markets, \\
and Rational Convex Programs}

\author[1]{Vijay V.~Vazirani\footnote{Supported in part by NSF grant CCF-2230414.}}

\affil[1]{University of California, Irvine}

\begin{document}
	\maketitle

\begin{abstract}
This paper unifies two foundational constructs from economics and algorithmic game theory—the Arctic Auction and the linear Fisher market—to address the efficient allocation of differentiated goods in complex markets. We focus on the Arctic Auction \cite{Klemperer2} that was developed for the Government of Iceland to allow individuals to exchange blocked offshore assets. It is a variant of the product-mix auction \cite{Klemperer-2008, Klemperer1, Klemperer2} that was designed for, and used by, the Bank of England, to allocate liquidity efficiently across banks pledging heterogeneous collateral of varying quality.

Our main contribution is to show that an equilibrium for the Arctic Auction is captured by a \emph{Rational Convex Program} (RCP)~\cite{Va.rational}, thereby revealing deep structural regularities between market design and convex optimization. Building directly on primal–dual techniques for the linear Fisher market~\cite{DPSV}, we derive the first combinatorial polynomial-time algorithm for computing Arctic Auction equilibria. 
\end{abstract}

%\today 

\bigskip
\bigskip
\bigskip
\bigskip
\bigskip
\bigskip
\bigskip
\bigskip
\bigskip
\bigskip
\bigskip
\bigskip
\bigskip
\bigskip
\bigskip
\bigskip
\bigskip
\bigskip
\bigskip

\pagebreak

%%%%%%%%%%%%%%%%%%%%

\clearpage
\pagenumbering{arabic}

\input{Intro}

\input{Related}

\input{RCP}

\input{Comb-Alg}

\input{proof.tex}

\input{Cost.tex}

\input{Step}

\input{Discussion}

\input{Ack}

	\bibliographystyle{alpha}
	\bibliography{refs}

\end{document}

%% file: Intro.tex
\section{Introduction}
\label{sec.intro}

%\paragraph{Arctic Auctions.}

This paper brings together two foundational notions from economics and algorithmic game theory (AGT), namely Arctic Auctions and linear Fisher markets, to address the challenge of efficiently allocating differentiated goods in complex markets. 

We focus on Klemperer's {\em Arctic Auction} \cite{Klemperer2}, a novel and powerful mechanism designed to solve resource allocation problems, notably those encountered in central banking. Our core technical contribution is to establish a fundamental link between this auction design and optimization theory: we show that the equilibrium of the Arctic Auction is captured by a {\em Rational Convex Program (RCP)} \cite{Va.rational}. Next, by building on primal-dual techniques developed by Devanur, Papadimitriou, Saberi, and Vazirani for the {\em linear Fisher market} \cite{DPSV}, we are able to derive the {\em first combinatorial polynomial-time algorithm} for computing an equilibrium for the Arctic Auction. 

%This work not only provides an exact computational solution for a mechanism of immediate practical importance but also solidifies the theoretical connection between modern market design and convex optimization. 

%Note that incentive compatibility is not expected in the Arctic Auction since it implements a competitive equilibrium rather than a dominant-strategy mechanism. 

\subsection{Product-Mix Auctions}
\label{sec.product}

The {\em Product-Mix Auction} of Klemperer \cite{Klemperer1, Klemperer2, Klemperer-2008} is a general sealed-bid mechanism that computes competitive-equilibrium allocations across multiple substitutable goods by solving an optimization that maximizes total surplus, where bidders submit mutually exclusive bids across product categories and the mechanism determines market-clearing prices and allocations simultaneously. The design of the Product-Mix Auction was motivated by the 2007--2008 financial crisis, when the Bank of England faced the challenge of allocating liquidity efficiently across banks pledging heterogeneous collateral of varying quality. The Arctic Auction is a version of the product-mix auction which handles settings where bidders may choose not to spend all their money if prices exceed their valuation thresholds. Variants of the Arctic Auction have been proposed by {\em IMF} staff as a tool for sovereign debt restructuring \cite{Willems-sovereign-debt}. 

For completeness we describe the connection between the Bank of England auction and Arctic Auction. In both auctions, bidders specify upper bounds on prices of goods they are interested in. Note that the ``goods'' may be different types of loans, the ``prices'' may be interest rates on these loans and the upper bounds reflect how much the bidder values each item. 

Under the Bank of England auction, a bid of $(60,50;3)$ means ``I want up to three units, and I value good $A$ at 60 per unit and good $B$ at 50 per unit.'' At auction prices $p_A$ and $p_B$, the bidder receives $A$ when $60 - p_A > 50 - p_B$ and $B$ when the reverse holds. In case of equality, the bidder gets 3 units of $A$ and $B$ combined, in any order. In contrast, an Arctic auction bid of $(60,50;500)$ means ``Spend my budget of 500 on whichever asset gives me better value, provided its price is no higher than my reservation price.'' The bidder chooses $A$ when 
\[{60 \over p_A} > {50 \over p_B}, \]
and if so receives the quantity $500/p_A$ of $A$. The bidder gets $B$ when the reverse holds. As described in Section \ref{sec.Prelim-Arctic}, this is precisely the way a Fisher market allocates goods. 
 
A defining feature of Product-Mix Auctions is that prices and the mix of products are determined together instead of running separate auctions for each product. Hence it accounts directly for bidders switching between products when relative prices change. Successful bidders pay the relevant market-clearing price, not their submitted bid price.

\subsection{Relationship between Fisher market and Arctic Auction}
\label{sec.Fisher-Arctic}

Equilibrium allocation and prices for the linear Fisher market model are captured by the celebrated Eisenberg-Gale convex program \cite{eisenberg}. Additionally, this was the first market model for which an exact, efficient, combinatorial algorithm was given \cite{DPSV}; henceforth we will call this the {\em DPSV algorithm}. The Arctic auction can be viewed as a quasi-linear Fisher market, in which buyers need not spend all their money and the utility of buyer $i$ for good $j$ gives an {\em upper bound on the price} at which $i$ will be happy to buy good $j$, i.e., if the price of $j$ equals this bound, $i$ is indifferent between getting money back or buying this good and if it exceeds this bound, $i$ will not buy this good \cite{Klemperer2}. 

We note that incorporating such upper bounds in the DPSV algorithm is far from straightforward for the following basic reason: The crucial property which yields efficient computability for a linear Fisher market is {\em weak gross substitutability}, i.e., if the demand of one good increases, the (equilibrium) price of another good cannot decrease. This property supports an algorithm which monotonically raises prices of goods, i.e., it never needs to reduce any price. Prices are dual variables of Eisenberg-Gale convex program, and it is well-known that primal-dual algorithms either raise or lower duals monotonically to achieve efficiency; schemes that raise {\em and} lower dual variables have not panned out in the past\footnote{The only exception we are aware of is Edmonds' weighted matching algorithm \cite{Edmonds.matching}.}. For the Arctic Auction, when the price of some good $j$ exceeds buyer $i$'s upper bound, $j$ should be returned to $i$. But in doing so, we will be forced to lower the price of good $j$, since its demand has decreased. The full picture is even more complicated: at any point in the algorithm, different goods have different prices, possibly unrelated to buyer $i$'s upper bounds, and for the same good, the upper bounds of different buyers are different. 

Notwithstanding these challenges, we proceeded to study a convex program appearing in Chen et al. \cite{Chen-Ye}. On applying KKT conditions, we realized not only does it capture equilibrium allocations and prices for the Arctic Auction but it is also an RCP, Theorems \ref{thm.arctic} and \ref{thm.RCP}, respectively. These are exceptionally strong structural properties, and there must be an underlying reason for their existence. That reason turned out to be the fact that in the quasi-linear Fisher market model, the {\em upper bound on the price} at which buyer $i$ will be happy to buy a good $j$ {\em is precisely her utility}, $u_{ij}$, for good $j$. In our algorithm, this equality is triggered by the event that $\alpha_i = 1$, where $\alpha_i$ is the {\em bang per buck} of buyer $i$, Definition \ref{def.bpb}. These facts dramatically changed the outlook---they brought hope despite the complicated picture painted above. 

The Icelandic government’s sovereign debt restructuring mechanism required computing competitive equilibria for a large number of possible supply scenarios and then selecting the most desirable outcome. This created a strong need for very fast equilibrium computation. The straightforward approach of using a convex program solver \cite{GLS, Vishnoi} proved too slow in practice \cite{Lock-implement}. Together with the fact that the convex program \eqref{eq.arctic} is an RCP, this motivated the search for a combinatorial polynomial-time algorithm for Arctic Auction.

A natural starting point was to extend the DPSV algorithm. However, that algorithm relies on the sophisticated machinery of balanced flows. Adapting these ideas to the Arctic Auction—specifically, returning money to buyers while preserving all feasibility and equilibrium conditions—turns out to be nontrivial; see Theorem~\ref{thm.return}.

%The Icelandic Government's sovereign debt restructuring mechanism resorted to computing competitive equilibria for many different supply possibilities and selecting the most desirable equilibrium, leading to a search for a fast way of computing these equilibria. The obvious approach of resorting to a convex program solver turned out to be too slow for this application \cite{Lock-implement}. This and the fact that convex program (\ref{eq.arctic}) is an RCP encouraged us to seek a combinatorial polynomial-time algorithm for Arctic Auction. It was natural to  extend the DPSV Algorithm. However, this algorithm uses the sophisticated machinery of {\em balanced flows}. Adapting these ideas to the Arctic Auction, i.e., arranging the return of  money to buyers without violating any of the conditions, is quite non-trivial, see Theorem \ref{thm.return}. 

This raises an interesting question: Is this positive outcome a one-off? I.e., can RCP and combinatorial polynomial-time algorithms go much further? We believe the answer is ``yes''. As evidence, in Section \ref{sec.cost} we study the Arctic Auction with constant marginal costs for goods on the supplier's side. We show that an optimal solution to this model is captured by an RCP and we give a simple combinatorial polynomial-time algorithm for it. 

Next, in Section \ref{sec.step} we study the Arctic Auction with stepwise increasing costs for goods on the supplier's side. This turns out the a useful cost function for the applications stated above \cite{Klemperer-Private}, see also \cite{Fichtl-computing}. We give a convex program, based on (\ref{eq.arctic}), and show that an optimal solution to it yields an equilibrium for this model. Additionally we show that this program (\ref{eq.step}) is also an RCP. Very recently, \cite{Arctic-Costs} obtained an efficient combinatorial algorithm for this market by building on the Arctic Auction algorithm presented in this paper.

%% file: Related.tex
\section{Related Works}
\label{sec.related}

The notion of Product-Mix Auction was given by Klemperer \cite{Klemperer1, Klemperer2}. The Arctic Auction was also introduced by Klemperer, as a specialized implementation of the product-mix auction, designed for central bank liquidity operations \cite{Klemperer2}. \cite{Fisher2011pricing} gave an early operational playbook showing how central banks price liquidity across collateral classes via Product-Mix Auctions. Fichtl \cite{Fichtl-computing} gave a proof of correctness of an algorithm proposed by Paul Klemperer and DotEcon for finding equilibrium prices for the Arctic Auction with stepwise increasing marginal costs. The algorithm performs well if the number of goods is very small, e.g., three. The algorithm was explicitly designed for solving a real practical problem in this regime; however it is exponential time in the worst case. It maximizes auctioneer surplus (a harder goal) as opposed to our algorithm which maximizes total surplus. 

\cite{Baldwin2019understanding} analyzed the existence of equilibrium in Product-Mix Auctions when goods are indivisible. They introduce the concept of "demand types" to characterize solutions when standard assumptions break down. \cite{BKL2024} developed polynomial-time algorithms to find competitive equilibrium prices and quantities for bidders in strong-substitutes product-mix auctions by using submodular minimization for price-finding and a novel constrained matching approach for supply allocation, directly utilizing the auction's explicit bidding language. \cite{Baldwin-Lang} established that the Product-Mix Auction's simple geometric bidding language uniquely represents all concave substitutes and strong-substitutes valuations, thereby providing a new characterization for these preference classes and ensuring the auction implements competitive equilibrium allocations.

\cite{Lock-Finster2023competitive} study pricing in quasi-linear, budget-constrained multi-good markets and show that the elementwise-minimal competitive equilibrium prices are simultaneously seller revenue maximizing and constrained welfare maximizing, a structural result obtained without assuming gross substitutability or homogeneity. Their work is complementary to ours: while they characterize the geometry and optimality of equilibrium price vectors in markets with budgets, our paper develops an algorithmic framework for computing such equilibria for the Arctic Auction. \cite{Lock-Finster2023competitive} also adapted a simplex-like algorithm by Adsul et al. \cite{Adsul-simplex} for Fisher markets to formulate the first combinatorial algorithm for the Arctic auction. However similar to \cite{Adsul-simplex}, their algorithm is unlikely to be  polynomial time.

In the conference version, \cite{DPSV-FOCS} proved that their algorithm for the linear Fisher market terminates with a rational equilibrium if all parameters in the given instance are rational numbers. A corollary of this result is that the Eisenberg-Gale convex program, which captures equilibrium allocations and prices for the linear Fisher market, always has rational optimal solutions. A direct proof of this fact was later given in \cite{Va.combinatorial}. 

Hence the Eisenberg-Gale convex program, a non-linear program, ``behaves'' like a linear program, an intriguing phenomenon indeed. Ye \cite{Ye-Rational} gave other convex programs which exhibit this behavior. Drawing analogy with a fundamental phenomenon in combinatorial optimization, namely the existence of LP-relaxations for some problems which always admit integral optimal solutions, e.g., for maximum matching and max-flow \cite{LP.book, Sch-book} (i.e., these LPs ``behave'' like integer programs), \cite{Va.rational} formally defined the notion of a Rational Convex Program (RCP), see Definition \ref{def.RCP}. The importance of this notion is that the existence of such a program gives a strong indication that the problem admits a combinatorial polynomial-time algorithm. 

RCPs have been found for several other market models, e.g., for the dichotomous-utilities case of the linear Arrow-Debreu Nash bargaining one-sided matching market \cite{GTV-MM}; 2-player Nash and nonsymmetric bargaining games \cite{Va.2-player-rational}; and a perfect price discrimination market model with production \cite{GoelV}. \cite{Devanur2016rational} gave an RCP for the linear Arrow--Debreu market. \cite{Cole-convex} gave new convex programs for natural generalizations of the linear Fisher market, establishing rationality of equilibria for some markets. In particular, they gave a convex program for the quasi-linear Fisher market, but they did not establish rationality for it; this program is related to the one studied in Section \ref{sec.RCP}. 

Several market equilibrium works utilize the flow-based technique introduced in \cite{DPSV}. Building on \cite{DPSV}, Orlin gave a strongly polynomial time algorithm for computing an equilibrium for the linear Fisher market \cite{Orlin-Fisher}. \cite{Duan2015combinatorial} gave a combinatorial polynomial algorithm for the linear Arrow-Debreu market and \cite{Duan-GM} gave an algorithm with an improved running time.  \cite{VY-2025} used these ideas to give a strongly polynomial algorithm for computing the Hylland-Zeckhauser equilibrium \cite{Hylland} for the case that each agent's utilities come from a bi-valued set. \cite{GTV-MM} gave an $\epsilon$-approximate equilibrium algorithm for the Arrow-Debreu extension of the classic Hylland-Zeckhauser mechanism \cite{Hylland} for a one-sided matching market. \cite{JV-Eisenberg} define the notion of Eisenberg–Gale markets and give efficient algorithms for several resource allocation markets in this class. \cite{Va.spending} gave a polynomial time algorithm for computing an equilibrium for the spending constraint market model. The following works have also used this technique \cite{Bei-computing, Bei-GH-ascending, Chaudhury-competitive}.

%% file: RCP.tex
\section{Preliminaries for Arctic Auctions}
\label{sec.Prelim-Arctic}

\begin{definition}
	\label{def.arctic}
(Klemperer \cite{Klemperer2})
An {\em Arctic Auction} is a variant of the linear Fisher market in which buyers need not spend all their money and the utility of buyer $i$ for good $j$ gives an {\em upper bound on the price} at which $i$ will be happy to buy good $j$. The market consists of a set $B$ of $n$ buyers, with buyer $i$ having money $m_i \in \Qplus$, and a set $G$ of $m$ divisible goods, w.l.o.g. one unit of each. The buyers have linear utilities for goods; assume $u_{ij} \in \Qplus$ is the utility of buyer $i$ for a unit of good $j$. If prices of all goods $j$  (weakly) exceed $u_{ij}$ then $i$ may be returned a part or all of her money; the money returned is denoted by $s_i$. If $x_{ij}$ is the amount of good $j$ allocated to buyer $i$, then her total utility from this bundle is:
\[ u_i(x, s_i) = \sum_{j \in G} {u_{ij} x_{ij}} + s_i  = w_i(x_i) + s_i, \] 
where $x_i$ is the restriction of $x$ to buyer $i$ and $w_i(x_i) = \sum_{j \in G} {u_{ij} x_{ij}}$. If $p_j$ is the price of one unit of good $j$ then the worth of this bundle is $\sum_{j \in G} {x_{ij} p_j} + s_i$.

We will make the mild assumption that each buyer has positive utility for some good and each good is desired by some buyer. Under these conditions, {\em equilibrium allocations and prices exist}, satisfying:
 \begin{enumerate}
 	\item Each buyer gets a {\em utility maximizing bundle}.
 	\item The {\em market clears}, i.e., all goods are sold and all money is either spent or returned. (In addition, all prices will be positive.) 
 \end{enumerate}
 \end{definition}
 
 \begin{note}
 	\label{note: rational}
 All parameters in all models defined in this paper, e.g., Definition \ref{def.arctic}, will be assumed to be non-negative rational numbers, i.e., belonging to $\Qplus$. 
 \end{note}

\begin{definition}
	\label{def.bpb}
W.r.t. prices $p$, define the {\em maximum bang per buck} of buyer $i$ to be 
\[ \alpha_i = \max_{j \in G} {\left\{ {u_{ij} \over {p_j}} \right\}}  . \]
Furthermore, if the above equality holds for good $j$, then we will say that $j$ is a {\em maximum bang per buck  good for $i$}. For the sake of succinctness we will say that $\alpha_i$ {\em  is the mbpb of buyer $i$ and $j$ is an mbpb good for $i$}.
\end{definition}
	
Since the utility function of each buyer is linear, it is easy to show that her utility maximizing bundle must contain only mbpb goods. Lemma \ref{lem.facts} follows easily from Definition \ref{def.arctic}.

\begin{lemma}
	\label{lem.facts}
Equilibrium allocations and prices for the Arctic Auction satisfy the following conditions for each buyer $i \in B$.  

\begin{enumerate}
	\item If $\exists j \in G$ s.t. $x_{ij} > 0$ then $j$ is an mbpb good for $i$. Furthermore, since $p_j \leq u_{ij}$ (because $u_{ij}$ is an upper bound on the price at which $i$ is happy to buy $j$), $\alpha_i \geq 1$.   
	\item If $s_i > 0$ and $\exists j \in G$ s.t. $x_{ij} > 0$ then $\alpha_i = 1$ and each good is bought by $i$ at the upper bound of its price.
	\item If  $\forall j \in G, \ x_{ij} = 0$ then $\forall j \in G, \ p_j \geq u_{ij}$, i.e., the price of each good is at least its upper bound, and hence $\alpha_i \leq 1$.      
\end{enumerate}
\end{lemma}

\section{A Rational Convex Program for Arctic Auction}
\label{sec.RCP}

We next study a variant of the Eisenberg-Gale convex program, given by Chen et al. \cite{Chen-Ye}. \cite{Chen-Ye} proved that an optimal solution to this program captures an equilibrium of the quasi-linear Fisher market; however, their proof does not clarify the central role of bang per buck of buyers, $\alpha_i$, as stated in Lemma \ref{lem.facts}. Since this is critical to our algorithm, we have provided it in Lemma \ref{lem.arctic}.

	\begin{maxi}
		{}  {(\sum_{i \in B}  {m_i (\log {(w_{i} (x_i)} + s_i)}) - s}  
			{\label{eq.arctic}}
		{}
		\addConstraint{\sum_{i \in B} {x_{ij}}}{\leq 1 \quad}{\forall j \in G}
		\addConstraint{{\sum_{i \in B} {s_{i}}} -s } {= 0}
		\addConstraint{x_{ij}}{\geq 0}{\forall i \in B, j \in G}
		\addConstraint{s_{i}}{\geq 0}{\forall i \in B}
	\end{maxi}

Let $p_j$ and $\lambda$ be the dual variables corresponding to the first and second constraints of (\ref{eq.arctic}).  Optimal solutions to the primal and dual variables must satisfy KKT conditions, in addition to the constraints of (\ref{eq.arctic}):

\begin{enumerate}
	\item $\forall j \in G: \ \ p_j \geq 0$.
	\item $\forall j \in G: \ \ p_j > 0 \ \ \implies \ \ \sum_{i \in B} {x_{ij}} = 1 .$ 
		\item $1 - \lambda \geq 0$.
	\item $s > 0 \ \ \implies \ \  \lambda = 1$ 
	\item $$ \forall \ i \in B,  \forall j \in G: \ \  {{u_{ij}} \over {p_j}} \leq {{w_i(x_i) + s_i} \over {m_i}} . $$
	\item $$ \forall \ i \in B,  \forall j \in G: \ \ x_{ij} > 0 \ \ \implies \ \ {{u_{ij}} \over {p_j}} = {{w_i(x_i) + s_i} \over {m_i}} . $$
	\item $$ \forall \ i \in B: \ \ \lambda \geq {{m_i} \over {w_i(x_i) + s_i}} .$$
	\item $$ \forall \ i \in B: \ \ s_i > 0 \ \ \implies \ \  \lambda = {{m_i} \over {w_i(x_i) + s_i}} .$$
\end{enumerate}

\begin{lemma}
	\label{lem.arctic}
Optimal primal and dual solutions to convex program (\ref{eq.arctic}) satisfy the statements given in Lemma \ref{lem.facts}. 
\end{lemma}

\begin{proof}

{\bf Statement 1:}  
Assume  $x_{ik} > 0$ for $k \in G$. Then by (6), 
\[ {{u_{ik}} \over {p_k}} = {{w_i(x_i) + s_i} \over {m_i}} . \]
By (5), 
\[ \forall j \in G: {{u_{ij}} \over {p_j}} \leq  {{w_i(x_i) + s_i} \over {m_i}} , \]
Therefore $k$ is an mbpb good for $i$ and 
\[ {{w_i(x_i) + s_i} \over {m_i}} = \alpha_i . \]
By (7), 
\[ {{w_i(x_i) + s_i} \over {m_i}} \geq  {1 \over \lambda} . \]
By condition (3), $1/\lambda \geq 1$. Hence we get that 
\[ \alpha_i \geq  {1 \over \lambda} \geq 1 .\]
Therefore if $x_{ij} > 0$, 
\[ {{u_{ij}} \over {p_j}} = \alpha_i \geq 1 . \]
Hence $p_j \leq u_{ij}$, i.e., the price of $j$ does not exceed $u_{ij}$.

\medskip

{\bf Statement 2:}  
Assume that $s_i > 0$ and $x_{ik} > 0$ for $k \in G$. Clearly $s > 0$, hence by condition (4), $\lambda = 1$. 

%Furthermore, and $k$ is an mbpb good therefore  
%\[ \alpha_i =  {{u_{ik} \over p_k}} . \]
By Statement (1) and condition (8),
\[ \alpha_i = {{w_i(x_i) + s_i} \over {m_i}} = {1 \over {\lambda}} = 1 . \] 
Now for each good $j$, if $x_{ij} > 0$ then $u_{ij}/p_j = 1$, i.e., $j$ is bought at the upper bound of its price.

{\bf Statement 3:}  
Assume that $\forall j \in G, \ x_{ij} = 0$. If so, $w_i(x_i) = 0$ and by (3) and (7), 
\[ {s_i \over m_i}  \geq {1 \over \lambda} \geq 1 .\] 
Therefore $s_i > 0$ and hence $s > 0$. Now by (4), $\lambda = 1$. Further, by (8), $s_i = m_i$. Now by (5), $\forall j \in G, \ p_j \geq u_{ij}$, i.e., the price of each good is at least its upper bound and hence $\alpha_i \leq 1$.  
\end{proof}

\begin{theorem}
	\label{thm.arctic}
Optimal primal and dual solutions to convex program (\ref{eq.arctic}) give equilibrium allocations and prices for the arctic auction. 
\end{theorem}

\begin{proof}
By Definition \ref{def.arctic}, two facts need to be established.  

{\bf 1). The market clears:} 
By the (mild) assumptions in Definition \ref{def.arctic}, each good $j$ has an interested buyer, hence $p_j > 0$. Now, by KKT condition (2), it is fully sold. Next we prove that each buyer $i$ spends all her money. There are three cases:

{\bf Case 1:} $s_i = 0$. The money spent by $i$ is
\[ \sum_{j \in G} {x_{ij} p_j} \ = \ 
\sum_{j \in G} {x_{ij}   {{m_i u_{ij}} \over {w_i(x_i)}} } \ = \ m_i . \]

{\bf Case 2:}
$s_i > 0$ and $\exists j \in G$ s.t. $x_{ij} > 0$. By Statement (2), $\alpha_i = 1$ and $m_i = w_i(x_i) + s_i$. Furthermore, since each good is bought by $i$ at the upper bound of its price, 
\[ w_i(x_i) = \sum_{j \in G} {u_{ij} x_{ij}} \ = \ \sum_{j \in G} {p_{j} x_{ij}} ,\]
i.e., the money spent on goods. Hence, $s_i$ is precisely the unspent money. 

{\bf Case 3:}
$\forall j \in G, \ x_{ij} = 0$. As shown in Statement (3), $s_i = m_i$. 

\bigskip  

 {\bf 2). Buyers get optimal bundles:}  
 By Statement (1) of Lemma \ref{lem.facts}, each good bought by $i$ is an mbpb good for $i$. Furthermore, if money is returned, in Case 2 (Case 3), it provides exactly (at least) as much utility as mbpb goods.  
\end{proof}

\bigskip 

\begin{definition}
	\label{def.RCP}
(Vazirani \cite{Va.rational})
A nonlinear convex program is said to be a {\em rational convex program (RCP)} if for any setting of its parameters to rational numbers such that it has a finite optimal solution, it admits an optimal solution that is rational and can be written using polynomially many bits in the number of bits needed to write all the parameters. (As stated in Note \ref{note: rational}, for all models studied in this paper, all parameters will be assumed to be non-negative rational numbers.) 
\end{definition}

\begin{theorem}
	\label{thm.RCP}
Convex program (\ref{eq.arctic}) is a rational convex program. 
\end{theorem}

\begin{proof}
Let us assume that all parameters, i.e., the $u_{ij}$s and $m_i$s, are rational numbers. We will show that equilibrium allocations and prices form a solution to a linear system hence proving the theorem.

We first give the variables of this linear system. For each good $j$, define a new variable, $q_j$, which is meant to be $1/p_j$. The linear system will solve for the $m$ $q_j$s, and the $p_j$s will be obtained by taking  reciprocals of $q_j$s. Among the $x_{ij}$s, guess the ones that are strictly positive. Assume there are $k$ such variables and discard the rest. Lastly, guess the variables $s_i$ that are strictly positive. Assume there are $l$ of these and discard the rest. 

Finally, we give the $m + k + l$ equations of the linear system over these $m + k + l$ variables. For each good $j$,  include the equality given in condition (2). For each positive $x_{ij}$, include the equality given in condition (6). For each positive $s_{i}$, include the equality given in condition (8) with $\lambda$ replaced by 1. Finally, if all $s_i$s are zero and $s = 0$, $\lambda$ plays a role only in condition (7) and this is easily satisfied by taking $\lambda$ to be the largest allowed value, namely $\lambda = 1$. 
\end{proof}

%% file: Comb-Alg.tex
\section{An Efficient Primal-Dual Algorithm for Arctic Auction} 
\label{sec.alg}

The fact that program (\ref{eq.arctic}) is a rational convex program, Theorem \ref{thm.RCP}, opens up the possibility of obtaining a polynomial time combinatorial algorithm for arctic auction.

\subsection{Balanced Flow in Network $\cN(p, r)$}
\label{sec.N}

Algorithm \ref{alg.one} builds on the polynomial time algorithm for a linear Fisher model given in \cite{DPSV}. In particular, it also critically uses the notion of {\em balanced flow} from \cite{DPSV}. Section 8 of the latter paper shows how to compute a balanced flow via $n$ max-flow computations. Remark \ref{rem.l2} explains the reason for using balanced flows and $\ell_2$ norm of the surplus vector. 

\begin{definition} 
\label{def.N}
Define $\cN(p, r)$ to be a directed network on vertices $B \cup G \cup \{s, t\}$, where $s$ and $t$ are the source and sink of $\cN(p, r)$. The network has two vectors as parameters, $p$ and $r$, where $p_j$ is the current price of good $j$ and $r_i$ is the {\bf money returned} to buyer $i$. The {\bf left-over money} of $i$ is defined to be $(m_i - r_i)$. The vectors $p$ and $r$ evolve as the algorithm proceeds, hence changing the network itself, as detailed in Algorithm \ref{alg.one}. The edges of $\cN(p, r)$ are:
\begin{enumerate}
	\item $\forall j \in G$: edge $(s, j)$ with capacity $p_j$. 
	\item $\forall i \in B$: edge $(i, t)$ with capacity $m_i - r_i$, i.e., the left-over money of $i$. 
	\item $\forall j \in G$ and $\forall i \in B$: if $j$ is an mbpb good for $i$, then the edge $(j, i)$ with capacity infinity. 
\end{enumerate}

If $f$ is a feasible flow in $\cN(p, r)$, its {\em value} is the amount of flow going from $s$ to $t$ and is denoted by $|f|$. For a vertex $v$ in $\cN(p, r)$, with $v \in (B \cup G)$, its {\bf degree} is defined to be the degree of $v$ in the subgraph of $\cN(p, r)$ induced on $B \cup G$ and is denoted by $\deg_{\cN(p, r)} (v)$. 
\end{definition}

\begin{definition}
	\label{def.balanced}
W.r.t. flow $f$ in network $\cN(p, r)$, for each buyer $i$ define $\surplus(i)$ to be $(m_i - r_i) - f(i, t)$, i.e., the left-over capacity of edge $(i, t)$. Let $\gamma(f)$ denote the vector of surpluses of all buyers. Then $f$ is said to be a {\em balanced flow} if it minimizes the $\ell_2$ norm of the surplus vector, i.e., $\norm{f}_2$, over all feasible flows $f$ in $\cN(p, r)$.  
\end{definition}

Clearly a balanced flow will be a max-flow in $\cN(p, r)$. Moreover, all balanced flows have the same surplus vector (Corollary 8.7 in \cite{DPSV}) and a balanced flow will make the surpluses of the buyers as equal as possible. Let $f$ be a balanced flow in $\cN(p, r)$ and let $\cR(f)$ be the corresponding residual network. The following property plays a key role in the DPSV algorithm as well as in Algorithm \ref{alg.one}.

{\bf Property 1:} Let $i, i' \in B$. If $\surplus(i') < \surplus(i)$ then there is no path from $i'$ to $i$ in $\cR(f)$. 

Theorem 8.4 in \cite{DPSV} shows that $f$ is a balanced flow iff Property 1 holds. The proof of the main direction is straightforward: if there were a path from $i'$ to $i$ in $\cR(f)$, we could find a circulation which increases $\surplus(i')$ and decreases $\surplus(i)$ thereby making the flow more balanced, leading to a contradiction.

\bigskip

%\setcounter{figure}{1} 
%\begin{figure}[H]

\begin{figure}

	\begin{wbox}
		\begin{alg}
		\label{alg.one}
		{\bf Algorithm for Arctic Auction}\\
		\\
			 {\bf Initialization} 
			 \begin{enumerate}
			 	\item $\forall j \in G$: $p_j \leftarrow \mn/m$. 
			 	\item $\forall i \in B$: $\alpha_i \leftarrow \max_{j \in G} {\left\{ {u_{ij} \over {p_j}} \right\}}$ \ \ \ \ and \ \ \ \ $r_i \leftarrow 0$.   
			 	\item Compute $\cN(p, r)$.
			 	\item $\forall j \in G$: If  $\deg_{\cN(p, r)} (j) = 0$ then  $p_j \leftarrow \max_{i \in B} {\left\{ {u_{ij} \over {\alpha_i}} \right\}}$.
			 	\item Recompute $\cN(p, r)$.
			 \end{enumerate}
		
		\medskip
		
			 {\bf New phase} 
			\begin{enumerate}
			\item $f \leftarrow$ balanced flow in $\cN(p, r)$. 
			\item $\forall i \in B$: $\surplus(i) \leftarrow (m_i - r_i) - f(i, t)$.
			\item $\delta \leftarrow \max_{i \in B} \{\surplus(i)\}$ . 
			\item $I \leftarrow \arg \max_{i \in B} {\{\surplus(i)\}}$ . 
			\item $J \leftarrow \Gamma(I)$.
			\item Remove edges going from $B - I$ to $J$.
			\item Let $Z = \{i \in (B-I) \ | \ \deg_{\cN(p, r)} (i) = 0 \}$. 
			\item Record initial prices of goods in $J$: $\forall j \in J$: $\overline{p}_j \leftarrow p_j$. 
			
		\medskip
		
			\item {\bf New iteration}
			\begin{enumerate}
			\item $\theta \leftarrow 1$.
			\item $\forall j \in J$: $p_j \leftarrow \overline{p}_j \cdot \theta$, \ \ \ \ $c(s, j) \leftarrow p_j$. 
			\item Raise $\theta$ continuously {\bf until}:
			\begin{enumerate}

			\item {\bf New edge:} If $\exists i \in I$ and $\exists j \in (G - J)$ s.t. $j$ is a mbpb good for $i$ then: 
			\begin{enumerate}
					\item Add $(i, j)$ with capacity infinity to $\cN(p, r)$. 
					\item Recompute balanced flow, $f$, in $\cN(p, r)$.
					\item $I' \leftarrow \{i \in (B-I) ~|~ \exists \mbox{path from $i$ to $I$ in $\CR(f)$} \}$ .
					\item $I \leftarrow (I \cup I')$.
					\item $J \leftarrow \Gamma(I)$. 
					\item Start a new iteration.
			\end{enumerate}
			\item {\bf A set goes tight:} If $\exists S \subseteq J: \worth(S) = \worth(\Gamma(S))$ then:  
			\begin{enumerate}
			\item If $S \neq G$, start a new phase.
			\item If $S = G$, Output current allocations and prices and {\bf Halt}.
			\end{enumerate}
			\item {\bf Money is returned to a buyer:} If $\exists i \in I$ s.t. $\alpha_i = 1$ then: 
			\begin{enumerate}
			\item Set $i$'s left-over money to 0. If $(s, G \cup B \cup t)$ is an $s$--$t$ min-cut, return all of $i$'s money, remove $i$ from $B$ and $\cN(p, r)$, and go to next phase. 
			\item If $(s, G \cup B \cup t)$ is not an $s$--$t$ min-cut, \\ find a maximal $s$--$t$ min-cut, say $(s \cup S \cup T, (G - S) \cup (B - T) \cup t)$. 
			\item Return $m_i - (\worth(S) - \worth(T))$  money to $i$. 
			\item Start a new phase. 
			\end{enumerate}
			
			\item {\bf Remove a buyer from $Z$:} If $\exists i \in Z$ s.t. $\alpha_i = 1$ then return all of $i$'s money, \\  remove $i$ from $B$ and $\cN(p, r)$, and continue raising $\theta$.
			\item {\bf New edge in $Z$:} If $\exists i \in Z$ and $j \in (G - J)$ s.t. $j$ is mbpb good for $i$ then: \\  add $(j, i)$ to $\cN(p, r)$, move $i$ from $Z$ to $(B - I) - Z$, and continue raising $\theta$.

			\end{enumerate}

			\end{enumerate}
			\end{enumerate} 

	%	\end{enumerate} 
		\bigskip
		\end{alg}
	\end{wbox}
\end{figure} 
%\bigskip

\subsection{Description of Algorithm \ref{alg.one}}
\label{sec.desc}

 Similar to \cite{DPSV}, Algorithm \ref{alg.one} is also a {\em primal-dual algorithm}. In our algorithm, the dual steps work on the dual variables, i.e., $p_j$s, by raising the price of a well-chosen set of goods. A max-flow in $\cN(p, r)$ determines the primal variables, i.e., $x_{ij}$s. The task of the primal steps is: 
 
 \begin{enumerate}
 	\item  Modify $\cN(p, r)$ appropriately so the subnetwork induced by $B \cup G$ contains all and only mbpb edges.
 	\item  Determine money returned, i.e., $r$. 
 \end{enumerate}

 Similar to \cite{DPSV}, Algorithm \ref{alg.one} maintains:

{\bf Invariant:} Algorithm \ref{alg.one} maintains the invariant that $(\{s\}, B \cup G \cup \{t\})$ is a minimum $s$--$t$ cut in $\cN(p, r)$. 

Because of this Invariant, the prices computed by the algorithm at any point are (weakly) bounded by equilibrium prices, i.e., the algorithm resorts to raising prices only, and not raising {\em and} lowering prices.

Let $\mn$ denote the minimum money of a buyer, i.e., $\mn = \min_{i \in B} {m_i}$. Step 1 of Algorithm \ref{alg.one} initializes the price of each good to $\mn/m$, ensuring that each buyer can buy all $m$ goods with her money. This ensures that $\alpha_i > 1$ and after Step 5, when $\cN(p, r)$ contains all mbpb edges, it satisfies the Invariant. Step 2 initializes the mbpb, $\alpha_i$, of each buyer $i$ and sets $r_i$ to 0.  At the beginning of each phase, a balanced flow is computed in $\cN(p, r)$ and the set $I \subseteq B$ of buyers having the largest surplus is identified (Step 4). $J \subseteq G$ consists of the set of goods desired by buyers in $I$. 

By Property 1, there is no residual path from a buyer in $(B - I)$ to a buyer in $I$. Therefore there is no flow going from goods in $J$ to buyers in $(B - I)$. In Step 6 of the New Phase, such edges are dropped. Let $Z \subseteq B - I$ be buyers all of whose mbpb edges are in $I$; clearly, after Step 6, such buyers will have no edges to $G$. We will call them {\bf zero-degree buyers} and will put them in the special set $Z$. For $i \in Z$, edge $(i, t)$ carries no flow and $\surplus(i) = m_i$. Importantly, the mbpb, $\alpha_i$ of such a buyer $i$ decreases as prices of goods in $J$ increase. In contrast, all mbpb goods of buyers in $(B - I) -Z$ are in $G - J$ and since the price of these goods remains unchanged, their mbpb remains unchanged. The algorithm keeps track of buyers in $Z$ for the following two events:
\begin{enumerate}
	\item If $\exists i \in Z$ s.t. $\alpha_i = 1$, Step 9(c)(iv) is triggered, which removes $i$ from $B$ and $\cN(p, r)$. Since edge $(i, t)$ carried no flow, the removal of $i$ does not violate the Invariant. Following this, the algorithm continues raising the prices of goods in $J$. Observe that until $i$'s money is returned, we cannot raise prices of goods in $J$. 
 
	\item If $\exists i \in Z$ and $j \in (G - J)$ s.t. $j$ becomes a mbpb good for $i$, Step 9(c)(v) is triggered, which adds $(i, j)$ to $\cN(p, r)$ with infinity capacity. $i$ is not a zero-degree buyer anymore and is moved from $Z$ to $(B - I) - Z$, and we continue raising $\theta$.
\end{enumerate} 

The main changes from \cite{DPSV} are {\em primal steps} Step 9(c)(iv), Step 9(c)(v) and Step 9(c)(iii). The last step is triggered when $\alpha_i$ of a buyer $i \in I$ becomes 1. In this case, either $i$ is removed after returning all her money or is retained after returning a part of her money.  

{\bf Notation} $\Gamma(S)$ will denote the {\em neighborhood} of set $S$ in $\cN(p, r)$ as follows: If $S \subseteq G$ then $\Gamma(S) \subseteq B$ and if $S \subseteq B$ then $\Gamma(S) \subseteq G$.  

We will define the {\em worth of a set} as follows: If $S \subseteq G$ then $\worth(S)$ is defined to be the sum of prices of goods in $S$ and if $S \subseteq B$ then $\worth(S)$ is defined to be the sum of leftover money of buyers in $S$. By the Invariant, $\forall S \subseteq G: \ \worth(S) \leq \ \worth(\Gamma(S))$. Note that while computing $\worth(T)$ in Step 9(c)(iii)(C), we will assume that the money of $i$ is zero. 

The algorithm runs in {\em phases}. A phase is partitioned into {\em iterations}. During each iteration, the algorithm increases the prices of goods in set $J$. For $j \in J$, $\overline{p}_j$ will denote the price of $j$ at the beginning of the iteration. Its current price, $p_j$, is $\overline{p}_j \cdot \theta$, where $\theta$ is initialized to 1 and is raised continuously at rate 1. As mentioned above, at the beginning of an iteration, edges from $J$ to $B - I$ will be dropped; by definition, there are no edges from $(G-J)$ to $I$. Hence the subnetwork induced on $(I, J)$ is decoupled from $(G-I, B-J)$. An iteration ends when a new edge enters $\cN(p, r)$, in Step 9(c)(i), or when the phase ends.

A phase ends when one of these events happens: 

\begin{enumerate}
	\item  A set $S \subseteq G$ goes {\em tight}, i.e., $\worth(S) = \worth(\Gamma(S))$, in Step 9(c)(ii).
	\item  $\alpha_i = 1$ for some buyer $i$ and her money is fully returned, in Step 9(c)(iii)(A).
	\item  $\alpha_i = 1$ for some buyer $i$ and her money is partially returned, in Step 9(c)(iii)(D).
\end{enumerate}

In the third possibility, in Step 9(c)(iii)(B), $(s, G \cup B \cup t)$ is not an $s$--$t$ min-cut in $\cN(p, r)$, indicating that returning all money to $i$ will violate the Invariant. Hence if $(s \cup S \cup T, (G - S) \cup (B - T) \cup t)$ is a maximal $s$--$t$ min-cut, then $\worth(S) > \worth(T)$. As mentioned above, while computing $\worth(T)$ in Step 9(c)(iii)(C), we will assume that the money of $i$ is zero. 

In Step 9(c)(iii)(C), the algorithm will return $r_i = m_i -  (\worth(S) - \worth(T))$ money to $i$. As a result, $(s, G \cup B \cup t)$ becomes an $s$--$t$ min-cut again $\cN(p, r)$, and the Invariant is restored, as shown in Lemma \ref{lem.Inv-partial}. Moreover, $S$ is now a tight set and therefore the phase ends. 

\begin{lemma}
	\label{lem.Inv-partial}
In Step 9(c)(iii)(C), after returning $m_i - (\worth(S) - \worth(T))$ money to $i$, $(s, G \cup B \cup t)$ becomes an $s$--$t$ min-cut again, thereby restoring the Invariant. At this point, $(s \cup S \cup T, (G - S) \cup (B - T) \cup t)$ is also an $s$--$t$ min-cut in $\cN(p, r)$, i.e., $S$ is a tight set. 
\end{lemma}

\begin{proof}
	Before Step 9(c)(iii)(A) is executed, the Invariant holds and therefore $(s, G \cup B \cup t)$ is an $s$--$t$ min-cut in $\cN(p, r)$. After Step 9(c)(iii)(A) is executed and all of $i$'s money is returned, $(s, G \cup B \cup t)$ is not an $s$--$t$ min-cut but $(s \cup S \cup T, (G - S) \cup (B - T) \cup t)$ is. At this point, let us continuously increase the money of buyer $i$ from zero until $(s, G \cup B \cup t)$ also becomes an $s$--$t$ min-cut. Let $\beta$ be the money of $i$ when this happens. Then clearly, $\worth(S) = \worth(T) + \beta$. Therefore on setting $r_i = m_i - (\worth(S) - \worth(T))$ the Invariant is restored and $S$ is a newly tight set. 
\end{proof}

Observe that at this point, $\alpha_i$ is still 1. Now there are two possibilities: First, $i$ enters $I$ in a future phase. At the start of that phase, Step 9(c)(iii) will kick in again since $\alpha_i = 1$. Again $i$ will either be returned all or part of her money. In the former case, the prices of $i$'s mbpb goods can be raised. After increasing prices, $\alpha_i < 1$, but this is fine since at these high prices, $i$ prefers her money back, and has already been arranged. In the latter case, $i$ is in a tight set again, and the whole process will repeat. 

Second, the algorithm terminates without $i$ entering the active set $I$ again. One possibility is that $i$ is removed in Step 9(c)(iv) with all her money returned. Otherwise at termination, in the final max-flow in Step 9(c)(ii)(B), $i$ will be allocated her mbpb goods worth her remaining money. Since $\alpha_i$ is still 1, $i$ is happy with a combination of money back and her mbpb goods. Note that the algorithm terminates when $(\{s\} \cup B \cup G, \{t\})$ becomes a min-cut.

If $\alpha_i > 1$ throughout the run of the algorithm, Step 9(c)(iii) will not kick in and no money will be returned to $i$. At termination, in Step 9(c)(ii)(B), $i$ will be allocated her mbpb goods worth her total money. Algorithm \ref{alg.one} raises prices monotonically and therefore for each buyer $i$, $\alpha_i$ decreases  monotonically. Combining with all the observations made above, including those for Steps 9(c)(iv) and 9(c)(v), we get Theorem \ref{thm.return}, proving correctness of Algorithm \ref{alg.one} on the main new feature of the Arctic Auction over and above a linear Fisher market. 

\begin{theorem}
	\label{thm.return}
For each buyer $i$, Algorithm \ref{alg.one} allocates a bundle of goods and money satisfying:
\begin{enumerate}
	\item If $\alpha_i > 1$, $i$ will be allocated her mbpb goods worth her entire money.
	\item If $\alpha_i = 1$, $i$ may be returned money and any goods allocated will be her mbpb goods. The total  worth of the bundle will be $m_i$.
	\item If $\alpha_i < 1$, $i$ will be returned all her money. 
\end{enumerate}
\end{theorem}

%% file: proof.tex
\section{Proof of Running Time}
\label{sec.time-proof}

As in \cite{DPSV}, our proof is based on the following potential function
\[
\Phi \;=\; \sum_{i\in B} \gamma_i^2, 
\]

where $\gamma_i = (m_i-r_i)-f(i,t)$, where 

\begin{itemize} 
\item \(m_i\) is buyer \(i\)'s original budget; 
\item \(r_i\) is the amount of money that has already been returned to buyer \(i\); and 
\item \(f(i,t)\) is the amount that buyer \(i\) currently spends under flow $f$.  
\end{itemize}

Clearly the potential decreases whenever money is returned to a buyer or the buyer spends more money. 
We will show that in a phase, this potential decreases by at least a multiplicative factor of 
\(
1 - 1/n^3.
\)
We will classify phases into three types based on how they terminate, and we quantify the progress made in each type of phase.

\begin{description}
    \item[Type I:] Phase ends in Step 9(c)(ii), when a set $S \subseteq J$ goes tight.
    \item[Type II:] Phase ends in Step 9(c)(iii)(A), when a buyer $i$ has all money returned and is removed.
    \item[Type III:] Phase ends in Step 9(c)(iii)(D), when buyer $i$ has partial money returned.
\end{description}

\medskip

\begin{remark}
	\label{rem.l2}
The main difficulty in establishing the running time of Algorithm \ref{alg.one}, as well as the DPSV algorithm, comes from the use of balanced flows and $\ell_2$ norm of the surplus vector. The question arises: Can we not argue about the $\ell_1$ norm of the surplus vector? It turns out there is an infinite family of instances for the DPSV algorithm for which the decrease in the $\ell_1$ norm of the surplus vector is inverse exponential. 

The next question is: what is the advantage of using $\ell_2$ norm? To answer this, consider two 2-dimensional vectors $(1, 0)$ and $(1/2, 1/2)$. The $\ell_1$ norm of both vectors is 1. On the other hand, the $\ell_2$ norms are $1$ and $1/\sqrt{2}$, respectively. In an iteration, we may go from the first to close to the second vector, resulting in almost no improvement in the $\ell_1$ norm-based potential function. On the other hand,  there is a substantial improvement in the $\ell_2$ norm-based potential function. Indeed, certain iterations of the algorithm do not lead to much decrease in the total surplus, but they make the surplus vector more balanced, thereby increasing the potential for a substantial improvement in future iterations.  
\end{remark}

Lemma \ref{lem.surplus-decrease}, which is based on Lemma 8.5 and Corollary 8.6 in \cite{DPSV}, plays a critical   role in \cite{DPSV} and our proof for showing that sufficient progress is made towards decreasing the $\ell_2^2$ norm of the surplus vector in an iteration. For this reason, we have provided a detailed proof, which is more intuitive and complete than Lemma 8.5 from \cite{DPSV}, and may be of independent interest.

\begin{lemma}
\label{lem.surplus-decrease}
Let $p$ and $p^{*}$ be the price vectors at the beginning and end of an iteration of the algorithm in which no money is returned (so $r$ is fixed). Let $f$ and $f^*$ be balanced flows in $\cN(p, r)$ and $\cN(p^*, r)$, respectively. Suppose that for some buyer $i$ and some $\sigma>0$ we have
\[
\gamma_i(f^{*}) \;=\; \gamma_i(f)-\sigma.
\]
Then
\[
\|\gamma(f^{*})\|_2^2 \;\le\; \|\gamma(f)\|_2^2 - \sigma^2.
\]
\end{lemma}

\begin{proof}
Since this iteration raises prices but never returns money, the capacities of edges from the source do not decrease, and the capacities into $t$ remain fixed. If this iteration ended in Step 9(c)(i), $N(p^*,r)$ contains an additional edge, compared to $N(p,r)$. Hence $f$ is feasible in $N(p^{*},r)$. Since $f^{*}$ is a balanced (and thus maximum) flow in $N(p^{*},r)$, we have $|f^{*}|\ge |f|$.

Consider the difference $f^{*}-f$ in the residual network of flow $f$ in $N(p^{*},r)$. This difference decomposes into a set of $s\!\to\!t$ augmenting paths of total value $|f^{*}|-|f|$ together with a circulation of value $|f|$. Select from this decomposition exactly those paths/cycles that traverse the edge $(i,t)$, and \emph{augment} $f$ along these to obtain a feasible flow $h$ in $N(p^{*},r)$. By construction, the flow on $(i,t)$ increases by exactly $\sigma$, and therefore
\[
\gamma_i(h) \;=\; \gamma_i(f)-\sigma \;=\; \gamma_i(f^{*}).
\]

Since $f^{*}$ is the balanced flow for $N(p^{*},r)$ and $h$ is a feasible flow in that network, we have
\[
\|\gamma(f^{*})\|_2 \;\le\; \|\gamma(h)\|_2.
\]
Thus it suffices to prove
\[
\|\gamma(h)\|_2^2 \;\le\; \|\gamma(f)\|_2^2 - \sigma^2.
\]

For establishing this inequality, we need to consider only buyers whose surplus increases in $h$ (other than $i$).
  For each such buyer $j$, let $\sigma_j\ge 0$ denote its surplus increase. Flow conservation in the selected augmentation implies $\sum_j \sigma_j \le \sigma$.

Furthermore, in $f^{*}$ there is a residual path from $i$ to each such $j$ (because the portions of $f^{*}-f$ used to form $h$ include these paths). Now Property~1 of balanced flows (no residual path from strictly smaller
surplus to strictly larger surplus), implies 
\[
\gamma_i(f^{*}) \;\ge\; \gamma_j(f^{*}).
\]
Since
$\gamma_i(h)=\gamma_i(f^{*})$
and
$\gamma_j(h)\le \gamma_j(f^{*})$,
we have for all affected $j$:
\[
\gamma_i(h) \;\ge\; \gamma_j(h).
\]

Finally, for concreteness, assume that there are $k$ such buyers $j$ namely $i_1, i_2, \ldots, i_k$ and positive values $\sigma_1, \ldots , \sigma_k$ with $\sum_{l=1}^k \sigma_l \leq \sigma$ such that for $1 \leq l \leq k$: 
\[ \gamma_i(h) = \gamma_i(f) - \sigma \geq \gamma_{i_l}(f) + \sigma_l = \gamma_{i_l}(h) \]  

Now using Lemma \ref{lem.DPSV} (i.e., which is Lemma 8.3 from \cite{DPSV}), we get 
\[
\|\gamma(f)\|_2^2 \;-\; \|\gamma(h)\|_2^2 \;\ge\; \sigma^2.
\]

Combining this with $\|\gamma(f^{*})\|_2\le\|\gamma(h)\|_2$ yields
\[
\|\gamma(f)\|_2^2 - \|\gamma(f^{*})\|_2^2 \;\ge\; \sigma^2,
\]
as required.

\end{proof}

\begin{lemma}
\label{lem.DPSV}
(Lemma 8.3 from \cite{DPSV})
Assume that $a \geq b_i \geq 0$ for $i = 1, 2,... , n$. Assume further that $\sum_{l=1}^k \sigma_l \leq \sigma$ , where $\sigma , \sigma_l \geq 0$. Then 
$$   {\|(a+ \sigma , b_1 - \sigma_1, b_2 - \sigma_2, \ldots  b_n -  \sigma_n)\|^2} \leq {\|{(a, b_1, b_2, \ldots , b_n)\|}^2}  - \sigma^2 . $$
\end{lemma}

Next we quantify the decrease in the potential $\Phi$ in a Type~I phase. 

\begin{lemma}
	\label{lem.No-I-phases}
	{\bf (Bounding Type I Phases)}
During any Type~I phase, the potential 
\[
\Phi \;=\; \sum_{i\in B} \gamma_i^2
\]
decreases by at least a multiplicative factor of 
\(
1 - 1/n^3.
\)
\end{lemma}

\begin{proof}
Let $p_0$ be the price vector at the start of the Type~I phase and let 
$f_0$ be the balanced flow in $N(p_0,r)$, where $r$ remains fixed throughout 
the phase.  Let
\[
\delta \;=\; \max_{i\in B} \gamma_i(f_0)
\qquad\text{and}\qquad
I_0 \;=\; \{\, i \in B : \gamma_i(f_0)=\delta \,\}.
\]
Let $J_0=\Gamma(I_0)$ denote the set of goods connected to buyers in $I_0$.

During a Type~I phase the algorithm repeatedly increases the prices of all 
goods in~$J$, and whenever a new maximum bang-per-buck edge enters the 
equality graph (Step~9(c)(i)), we update $I$ and $J$ and recompute a 
balanced flow, thereby starting a new \emph{iteration} of the phase.
Let these iterations be $0,1,\dots,k$. 
Since each such iteration adds at least one new good to~$J$ and $|J|\le n$,
we have $k\le n$.

For each iteration $t$, let $p_t$ be the price vector and let $I_t$ be the 
corresponding set of active buyers.  Define
\[
\delta_t \;=\; \min_{i\in I_t} \gamma_i(p_t).
\]
At the start of the phase $\delta_0=\delta$, and at the end of the phase 
(event~(1) of Step~9(c)) we have $\delta_k=0$ because some buyer in~$I_k$ has 
surplus~$0$.

Since $\delta_0=\delta$ and $\delta_k=0$ and there are at most $k\le n$ 
iterations, there exists an iteration $t\in\{1,\dots,k\}$ such that 
\[
\delta_{t-1}-\delta_t \;\ge\; \frac{\delta}{n}.
\tag{1}
\label{eq:big-drop}
\]

We next show that this drop occurs in the surplus of a single buyer already in 
$I_{t-1}$.  
If $i'\in I_t\setminus I_{t-1}$ is a newly added buyer, then by Step~9(c)(i)(C)
there is a path in the residual graph of the balanced flow (for prices $p_t$)
from $i'$ to some buyer $h\in I_{t-1}$.  
By Property~1 of balanced flows, which prohibits residual paths from lower 
surplus to higher surplus buyers, we obtain 
$\gamma_{i'}(p_t)\ge \gamma_h(p_t)\ge\delta_{t-1}$.  
Hence $\delta_t$ cannot be achieved by a newly added buyer, and therefore 
there exists $i\in I_{t-1}$ such that 
\[
\gamma_i(p_{t-1}) - \gamma_i(p_t) \;\ge\; \delta_{t-1}-\delta_t 
\;\ge\; \frac{\delta}{n}.
\tag{2}
\label{eq:sigma}
\]

Let $\sigma = \gamma_i(p_{t-1}) - \gamma_i(p_t)$; then $\sigma \ge \delta/n$.  
Now we apply Lemma \ref{lem.surplus-decrease} to balanced flows $f_{t-1}$ and $f_t$ at prices $p_{t-1}$ and 
$p_t$. If the surplus of some buyer decreases by $\sigma$, then
\[
\|\gamma(f_{t-1})\|_2^2 \;-\; \|\gamma(f_t)\|_2^2 
\;\ge\; \sigma^2.
\tag{3}
\label{eq:lemma4}
\]

Therefore, from iteration $t-1$ to iteration~$t$ we have
\[
\Phi_{t-1} - \Phi_t \;\ge\; \sigma^2 \;\ge\; \frac{\delta^2}{n^2}.
\tag{4}
\label{eq:phase-drop}
\]

In all other iterations of the phase, the potential $\Phi$ is nonincreasing:
when prices rise, we recompute a balanced flow at the new prices, resulting 
in a surplus vector with $\ell_2$-norm no larger than before.  
Thus overall, from the start to the end of the Type~I phase,
\[
\Phi_{\mathrm{after}} 
\;\le\; \Phi_{\mathrm{before}} - \frac{\delta^2}{n^2}.
\tag{5}
\label{eq:total}
\]

Finally, at the start of the phase every buyer has surplus at most $\delta$,
so
\[
\Phi_{\mathrm{before}} 
\;=\; \sum_{i\in B} \gamma_i^2 
\;\le\; n\,\delta^2.
\]
Combining this with~\eqref{eq:total} gives
\[
\Phi_{\mathrm{after}}
\;\le\; 
\Phi_{\mathrm{before}}
\;-\; \frac{\delta^2}{n^2}
\;\le\; 
\Phi_{\mathrm{before}}\left(1-\frac{1}{n^3}\right),
\]
which completes the proof.
\end{proof}

\bigskip
\bigskip

\begin{lemma}
\label{lem.TypeII-count}
	{\bf (Bounding Type II Phases)}
There are at most $n$ Type II phases throughout the algorithm.
\end{lemma}

\begin{proof}
Each Type II phase removes one buyer from $B$. Since we start with at most $n$ buyers, there can be at most $n$ such phases.
\end{proof}

\bigskip

We next consider Type III phases which end in Step 9(c)(iii)(D), i.e., when buyer $i$ has partial money returned.

\begin{lemma}
	\label{lem.No-III-phases}
		{\bf (Bounding Type III Phases)}
During any Type~III phase, the potential $\Phi$ decreases by at least a multiplicative factor of 
\(
1 - 1/(4n^3).
\)
\end{lemma}

\begin{proof}
We will use the notation set up in the proof of Lemma \ref{lem.No-I-phases}. Let $\delta$ be the maximum surplus of a buyer at the start of this phase. There are $k \leq n$ iterations, with iteration $k$ ending in the return of money $\rho_i$ to buyer $i$. Let $r$ be the vector of money returned at the start of the phase and $r + \rho_i$ at the end. Now there are two cases:

{\bf Case 1:} There is a buyer $j \in I_{k-1}$ whose surplus drops to $\leq \delta/2$. \\
If so, there exists an iteration $t\in\{1,\dots,k-1\}$ such that 
\[
\delta_{t-1}-\delta_t \;\ge\; \frac{\delta}{2n}.
\tag{1}
\label{eq:big-drop4}
\]

As in Lemma \ref{lem.No-I-phases}, applying Lemma \ref{lem.surplus-decrease}, the desired result follows.

\bigskip

{\bf Case 2:} The minimum surplus of a buyer in $I_{k-1}$ is $> \delta/2$. \\
Since $i$ is in a tight set at the end of the phase, its surplus drops to zero. Let $f^*$ be a balanced flow in $\cN(p^*, r)$ when $\alpha_i$ becomes 1 and before money is returned to $i$. Let the drop in the surplus of $i$ up to this point in iteration $k$ be $\delta_k$ (note that $\delta_k$ may be $<0$). When money is returned to $i$, her surplus drops further by $\rho_i$, therefore, 
\[ \delta_k + \rho_i > \delta/2 .\] 
Hence one of $\delta_k$ and $\rho_i$ is at least $\delta/4$. In the first case, the argument of Lemma \ref{lem.No-I-phases} applies and gives the desired result. In the second case, on returning money, the surplus of $i$ drops by at least $\delta/4$ and the surplus of other buyers remains unchanged, resulting in a drop in the $\ell_1$ norm of the surplus vector by the same amount. Consequently, the $\ell_2^2$ norm of the surplus vector drops by at least $(\delta^2/4)$ and it may decrease further on recomputing a balanced flow. Since $\Phi$ at the start of the phase is $\leq n \delta^2$, we get: 
\[
\Phi_{\mathrm{after}}
\;\le\; 
\Phi_{\mathrm{before}} - \left(\frac{\delta}{4}\right)^2 
\leq 
\Phi_{\mathrm{before}}\left(1-\frac{1}{16n}\right)
\;\le\; 
\Phi_{\mathrm{before}}\left(1-\frac{1}{4n^3}\right),
\]
which completes the proof.

\end{proof}

\bigskip
\bigskip

We now prove that Algorithm \ref{alg.one} runs in time polynomial in $n$, $\log U$, and the
bit-size of the initial money vector $(m_i)$.  Our analysis mirrors that of
DPSV~\cite{DPSV}, but includes two differences specific to the Arctic Auction:
(i) buyers may have money returned during the algorithm, and hence their
effective budgets $(m_i - r_i)$ decrease over time; and (ii) the network
$N(p,r)$ contains buyer–return arcs that must be accounted for in the bit-size
analysis.  Importantly, the return of money never increases any coordinate of
the surplus vector, and therefore never increases the potential function $\Phi$.

\begin{theorem}
	\label{thm.time}
	{\bf (Running Time)} \\
Let all numbers $m_i$ and $u_{ij}$ be nonnegative rationals whose numerators
and denominators are written with at most $L$ bits in total, and let
$M=\sum_{i\in B} m_i$ and $U=\max_{i,j} u_{ij}$.  
Then Algorithm \ref{alg.one} terminates after at most
\[
O\!\Bigl(n^3\bigl(\log n + n\log U + \log M\bigr)\Bigr)
\]
phases, and performs at most
\[
O\!\Bigl(n^5\bigl(\log n + n\log U + \log M\bigr)\Bigr)
\]
max-flow computations. 
\end{theorem}

\begin{proof}

\medskip\noindent
At the start of the algorithm, $r_i=0$ for all $i$, so
$0 \le \gamma_i \le m_i$.  Thus
\[
\Phi_{\mathrm{init}}
\;=\;
\sum_{i} \gamma_i^2
\;\le\;
\sum_i m_i^2
\;\le\;
M^2.
\]
Whenever money is returned, $(m_i-r_i)$ only decreases.  Therefore $\gamma_i$
never increases, and hence
\[
\Phi \;\le\; \Phi_{\mathrm{init}} \;\le\; M^2
\quad\text{throughout the algorithm}.
\]
Thus decreasing money does not worsen the upper bound for Arctic Auction, compared to \cite{DPSV}.

\medskip\noindent
\textbf{Potential drop per phase:}
By Lemma \ref{lem.No-I-phases}, in every Type~I phase,
\[
\Phi_{\mathrm{after}} \;\le\;
\Bigl(1 - \frac{1}{n^3}\Bigr)\Phi_{\mathrm{before}}.
\]
By Lemma \ref{lem.No-III-phases} in every Type~III phase,
\[
\Phi_{\mathrm{after}} \;\le\;
\Bigl(1 - \frac{1}{4n^3}\Bigr)\Phi_{\mathrm{before}}.
\]
By Lemma \ref{lem.TypeII-count}, each Type~II phase removes one buyer permanently, and hence there
are at most $n$ Type~II phases.

Thus for every Type~I or Type~III phase, we have the uniform multiplicative decrease
\begin{equation}
\Phi_{\mathrm{after}}
\;\le\;
\Bigl(1 - \frac{1}{4n^3}\Bigr)\Phi_{\mathrm{before}}.
\label{eq:phase-drop2}
\end{equation}

\medskip\noindent
\textbf{Lower bound on nonzero potential:}
All prices, flows, surpluses, and money-return values arise as solutions of
linear systems of equations with coefficients taken from $\{m_i\}$, $\{u_{ij}\}$
and previous values of $(p,r)$, all of which have bit-size at most
$O(\log n+n\log U+\log M)$ by the same argument used in \cite{DPSV}.
Thus every rational number appearing in the algorithm has denominator of size at
most $2^{O(\log n+n\log U+\log M)}$.

Hence if any $\gamma_i > 0$, then
\[
\gamma_i \;\ge\; 
2^{-O(\log n+n\log U+\log M)},
\]
and therefore
\begin{equation}
\Phi \;\ge\; 
2^{-O(\log n+n\log U+\log M)}.
\label{eq:phi-min2}
\end{equation}
This lower bound does \emph{not} depend on the current total money
$\sum_i (m_i - r_i)$ and therefore remains valid even as money is returned.

\medskip\noindent
\textbf{Bounding the total number of phases and max-flow computations:}
Let $K$ denote the number of Type~I plus Type~III phases.
By repeated application of~\eqref{eq:phase-drop2}, we obtain
\[
\Phi_{\mathrm{final}}
\;\le\;
\Bigl(1 - \frac{1}{4n^3}\Bigr)^K \Phi_{\mathrm{init}}
\;\le\;
\Bigl(1 - \frac{1}{4n^3}\Bigr)^K M^2.
\]
As long as the algorithm has not terminated, $\Phi>0$, and hence by
\eqref{eq:phi-min2},
\[
2^{-O(\log n+n\log U+\log M)}
\;\le\;
\Bigl(1 - \frac{1}{4n^3}\Bigr)^K M^2.
\]
Taking logarithms and using $\ln(1-x)\le -x$, we obtain
\[
K
\;=\;
O\!\Bigl(n^3\bigl(\log n + n\log U + \log M\bigr)\Bigr).
\]

Each phase consists of at most $n$ iterations, since each iteration adds a new
mbpb edge and no such edge can be added twice. Since the balanced-flow procedure
of DPSV requires $n$ max-flow computations (Section~8 of~\cite{DPSV}), each iteration requires at most $O(n)$
max-flow computations.  Therefore each phase uses $O(n^2)$ max-flow calls, giving a total
of $O(n^5(\log n+n\log U+\log M))$ max-flows. 
\end{proof}

%% file: Cost.tex
\section{Arctic Auction with Linear Costs}
\label{sec.cost}

\begin{definition}
	\label{def.cost}
The {\em Arctic Auction with Linear Costs} differs from Definition \ref{def.arctic} in two respects: first, there is no upper bound on the amount of goods available for sale, and second, for each good $j$, the production cost incurred by the seller for good $j$ is given by a function $c_j: \Re_+ \rightarrow \Re_+$ with $c_j(y_j) = d_j \cdot y_j$, where $d_j > 0$ is a rational parameter. Because supply is unbounded, we are not seeking a {\em competitive} equilibrium.

Allocations and prices must satisfy: each buyer gets a utility maximizing bundle of goods and the money of each buyer is either spent on goods or is returned. The seller's {\em revenue} is defined to be 
\[ \sum_{i \in B} {m_i}  \ \ - s .     \] 
For each good $j$, let $y_j = \sum_{i \in B} {x_{ij}}$, i.e., the amount of good $j$ sold. Then the {\em profit}  of the seller is defined to be:
\[ \sum_{i \in B} {m_i}  \ \ - s \ \ - \sum_{j \in G} {c_j(y_j)}.     \] 
 \end{definition}

\bigskip

We will show that the following variant of program (\ref{eq.arctic}) captures optimal allocations and prices for the market given in Definition \ref{def.cost}.

	\begin{maxi}
		{}  {\sum_{i \in B}  {m_i (\log {(w_{i} (x_i)} + s_i)} - s - \sum_{j \in G} {d_j \cdot y_j}}  
			{\label{eq.cost}}
		{}
		\addConstraint{{\sum_{i \in B} {x_{ij}}} -y_j } {= 0 \quad}{\forall j \in G}
		\addConstraint{{\sum_{i \in B} {s_{i}}} -s } {= 0}
		\addConstraint{x_{ij}}{\geq 0 \quad}{\forall i \in B, j \in G}
		\addConstraint{s_{i}}{\geq 0 \quad}{\forall i \in B}
	\end{maxi}

Let $p_j$ and $\lambda$ be the dual variables corresponding to the first and second constraints of (\ref{eq.cost}).  Optimal solutions to the primal and dual variables must satisfy KKT conditions, in addition to the constraints of (\ref{eq.cost}):

\begin{enumerate}
	\item $\forall j \in G: \ \ d_j - p_j \geq 0$.
	\item $\forall j \in G: \ \ y_j > 0 \ \ \implies \ \ d_j = p_j  .$ 
		\item $1 - \lambda \geq 0$.
	\item $s > 0 \ \ \implies \ \  \lambda = 1$ 
	\item $$ \forall \ i \in B,  \forall j \in G: \ \  {{u_{ij}} \over {p_j}} \leq {{w_i(x_i) + s_i} \over {m_i}} . $$
	\item $$ \forall \ i \in B,  \forall j \in G: \ \ x_{ij} > 0 \ \ \implies \ \ {{u_{ij}} \over {p_j}} = {{w_i(x_i) + s_i} \over {m_i}} . $$
	\item $$ \forall \ i \in B: \ \ \lambda \geq {{m_i} \over {w_i(x_i) + s_i}} .$$
	\item $$ \forall \ i \in B: \ \ s_i > 0 \ \ \implies \ \  \lambda = {{m_i} \over {w_i(x_i) + s_i}} .$$
\end{enumerate}

%%%%%%%%%%%%%%%%%%%%%%%%

It is easy to see that the statements given in Lemma \ref{lem.facts} hold for this market as well and optimal primal and dual solutions to convex program (\ref{eq.cost}) satisfy them via the same proof as given in Lemma \ref{lem.arctic}.

\begin{theorem}
	\label{thm.cost}
There exist optimal primal and dual solutions to convex program (\ref{eq.cost}) which form a solution to a Fisher market with constant marginal costs for production. At such a solution, the seller's revenue is maximized and his profit will be zero. 
\end{theorem}

\begin{proof}
Three facts need to be established.  

 {\bf 1). Buyers get optimal bundles:}  
 By Lemma \ref{lem.facts}, each good bought by $i$ is an mbpb good for $i$, and   if money is returned, it provides at least as much utility as mbpb goods. Furthermore, for each buyer $i$, all her money is either spent on mbpb goods or is returned, i.e., 
 \[ m_i =  \sum_{j \in G} {u_{ij} x_{ij}} + s_i .\]

{\bf 2). The seller maximizes her revenue:} 
We will consider three cases:

{\bf Case 1:} $\alpha_i > 1$. If so, $i$'s mbpb goods give her more utility than money. Therefore in every optimal primal, $s_i = 0$, i.e., $i$ spends all her money on goods and it becomes the seller's revenue.

 {\bf Case 2:} $\alpha_i = 1$. If so, $i$ gets the same utility from her mbpb goods as from  money. In this case, we will assume that $i$ spends all her money on goods and it becomes the seller's revenue.
 
{\bf Case 3:} $\alpha_i < 1$. If so, $i$ derives more utility from money than her mbpb goods. Therefore in every optimal primal, $s_i = m_i$. 

{\bf 3). The seller's profit will be zero:} By KKT condition (2), if $y_j > 0$, $p_j = d_j$. As a result, all the revenue made from selling good $j$, i.e., $p_j \cdot y_j$, will be used for paying cost, since $d_j \cdot y_j = p_j \cdot y_j$. Therefore seller's profit will be zero.
\end{proof}

\bigskip

\begin{theorem}
	\label{thm.cost-RCP}
Convex program (\ref{eq.cost}) is a rational convex program. 
\end{theorem}

\begin{proof}
The proof is easier than the one for Theorem \ref{thm.RCP} because for each good $j$ that is sold to a non-zero extent, $p_j = d_j$, establishing its rationality. Corresponding to non-zero $x_{ij}$s and $s_i$s, we will use the linear equations given in Theorem \ref{thm.RCP} to obtain the required linear system, hence proving existence of an equilibrium in which these variables are rational. 
\end{proof}

\subsection{Efficient Greedy Algorithm}
\label{sec.cost-alg}

The algorithm for obtaining a solution is straightforward. We set the price $p_j$ of each good $j$ to $d_j$. If for a buyer $i$, $\alpha_i < 1$, i.e., the price of each good $j$ is greater than her upper bound of $u_{ij}$, then we return money $m_i$ to $i$ and remove her from consideration. 

For the remaining buyers, $\alpha_i \geq 1$. For each such buyer $i$, we find $M_i$, the set of her mbpb goods. Allocate goods $j \in M_i$ worth $m_i$ money arbitrarily, i.e., $x_{ij} > 0$ only if $j \in M_i$ and $\sum_{j \in G} {x_{ij} \cdot p_i} = m_i$. These prices and allocations will constitute a solution.

%% file: Step.tex
\section{Arctic Auction with Stepwise Increasing Marginal Costs}
\label{sec.step}

\begin{definition}
	\label{def.step}
An {\em Arctic Auction with stepwise increasing marginal costs} differs from Definition \ref{def.cost} in that the seller's cost function, $c_j(y_i)$ for each good $j$, is convex, piecewise-linear. Formally, for {\em breakpoints}
        \(0=b_{j0}<b_{j1}<\cdots<b_{jK_j}\) and {\em marginal costs}
        \(0\le d_{j1}<d_{j2}<\cdots<d_{jK_j}\),
         define {\em segment lengths} $l_{jk}=b_{jk}-b_{j,k-1}$. Then the  cost of producing $y_j\in[0,b_{jK_j}]$ units is
\[
c_j(y_j)=
\begin{cases}
d_{j1}\,y_j, & 0 \le y_j \le b_{j1},\\[4pt]
\displaystyle \sum_{k=1}^{m-1} d_{jk}\,l_{jk} \;+\; d_{jm}\,(y_j-b_{j,m-1}),
& b_{j,m-1}< y_j \le b_{jm},
\end{cases}
\]
for the unique $m\in\{1,\dots,K_j\}$ with $b_{j,m-1}<y_j\le b_{jm}$. The
(right) derivative is the step function $c_j'(y_j)$ on the open
interval $(b_{j,m-1},b_{jm})$, and at a breakpoint $b_{jk}$ the
subgradient is $\partial c_j(b_{jk}) \in [d_{jk},d_{j,k+1}]$. We have assumed that $b_{jK_j}$ is large enough that  the amount of good $j$ sold will not exceed this bound\footnote{One way of ensuring this is that for each good $j$, the total worth of good $j$ available is at least the total money of all buyers, i.e., $b_{jK_j} \cdot d_{jK_j} \geq \sum_{i \in B} {m_i}$.}. 
 \end{definition}

\bigskip

We will show that the following variant of program (\ref{eq.cost}) captures equilibrium allocations and prices for an arctic auction with step-increasing marginal costs.

	\begin{maxi}
		{}  {\sum_{i \in B}  {m_i (\log {(w_{i} (x_i)} + s_i)} - s - \sum_{j \in G} {c_j(y_j)}}  
			{\label{eq.step}}
		{}
		\addConstraint{{\sum_{i \in B} {x_{ij}}} -y_j } {= 0 \quad}{\forall j \in G}
		\addConstraint{{\sum_{i \in B} {s_{i}}} -s } {= 0}
		\addConstraint{x_{ij}}{\geq 0 \quad}{\forall i \in B, j \in G}
		\addConstraint{s_{i}}{\geq 0 \quad}{\forall i \in B}
	\end{maxi}

Let $p_j$ and $\lambda$ be the dual variables corresponding to the first and second constraints of (\ref{eq.step}). Optimal solutions to the primal and dual variables must satisfy KKT conditions, in addition to the constraints of (\ref{eq.step}). One KKT condition is the following:

\[ (*) \ \ \ \  \forall j \in G: \ p_j \in \partial c_j(y_j), \] 
where \(\partial c_j(y_j)\) denotes the subgradient of the function \(c_j(y_j)  \), i.e., if \(y_j\) lies strictly within a segment with slope
        \(d_{jk}\), then \(p_j=d_{jk}\), and if \(y_j\) is at a kink between slopes
        \(d_{jk}\) and \(d_{j,k+1}\),
        then \(p_j\in[d_{jk}, d_{j,k+1}]\). 
        
In addition, KKT conditions (3) to (8) for convex program (\ref{eq.cost}) will carry over verbatim.

Once again, it is easy to see that the statements given in Lemma \ref{lem.facts} hold for this market as well and optimal primal and dual solutions to convex program (\ref{eq.step}) satisfy them via the same proof as given in Lemma \ref{lem.arctic}. The proof of Theorem \ref{thm.Step} is along the same lines as the first two parts of the proof of Theorem \ref{def.cost} and is omitted.

\begin{theorem}
	\label{thm.Step}
There exist optimal primal and dual solutions to convex program (\ref{eq.step}) which form equilibrium allocations and prices for a Fisher market with step-increasing marginal costs.
\end{theorem}

\begin{theorem}
	\label{thm.step-RCP}
Convex program (\ref{eq.step}) is a rational convex program. 
\end{theorem}

\begin{proof}
The proof is similar to the one for Theorem \ref{thm.RCP}. For each good $j$ that is sold to a non-zero extent, if \(y_j\) lies strictly within a segment with slope
        \(d_{jk}\), then \(p_j=d_{jk}\), a rational number, and if $y_j$ is in the {\em kink} between two segments with slopes $d_{jk}$ and $d_{j, k+1}$, then $y_j = b_{jk}$, and $p_j$ lies in the interval \([d_{jk}, d_{j,k+1}]\). 
        
As in Theorem \ref{thm.RCP}, we will show that equilibrium allocations and prices form a solution to a linear system hence proving the theorem. We first give the variables of this linear system. For each good $j$ such that $y_j = b_{jk}$, i.e., $y_j$ lies in the kink between two segments, define a new variable, $q_j$, which is meant to be $1/p_j$. Guess these good $j$ and assume that are $m'$ such goods. The linear system will solve for the $m'$ $q_j$s, and the $p_j$s will be obtained by taking reciprocals of $q_j$s. For each such $j$, the linear system will have the equation $\sum_{i \in B} {x_{ij}} = b_{jk}$. 

Among the $x_{ij}$s, guess the ones that are strictly positive. Assume there are $k$ such variables and discard the rest. Lastly, guess the variables $s_i$ that are strictly positive. Assume there are $l$ of these and discard the rest. For each positive $x_{ij}$, include the equality given in condition (6). For each positive $s_{i}$, include the equality given in condition (8) with $\lambda$ replaced by 1. Finally, if all $s_i$s are zero and $s = 0$, $\lambda$ plays a role only in condition (7) and this is easily satisfied by taking $\lambda$ to be the largest allowed value, namely $\lambda = 1$. Hence we get a linear system involving $k + l + m'$ equations over $k + l + m'$ variables, completing the proof. 
\end{proof}

%% file: Discussion.tex
\section{Discussion}
\label{sec.Discussion}

Following up on \cite{DPSV}, Orlin gave a strongly polynomial time algorithm for computing an equilibrium for the linear Fisher market \cite{Orlin-Fisher}. Very recently such an extension was given for the Arctic auction as well \cite{Arctic-strongly}. 

%% file: Ack.tex
\section{Acknowledgements}

I wish to express my gratitude to Paul Klemperer and Edwin Lock for their generous help which got me started on this project. Paul introduced me to the elegant notion of Arctic Auction and Edwin informed me about several useful references, most importantly \cite{Chen-Ye}, and provided valuable insights into convex program (\ref{eq.arctic}) appearing in the latter paper. Additionally, both Paul and Edwin provided useful feedback on the first version of this paper.

%Help from AI tools is also acknowledged. 

\section{Declaration of generative AI and AI-assisted technologies in the manuscript preparation process}

During the preparation of this work the author used Chat GPT 5 in order to complete the proof of the running time of the algorithm. After using this tool, the author reviewed and edited the content as needed and takes full responsibility for the content of the published article.